\documentclass[runningheads]{llncs}
\usepackage{graphicx}
\usepackage{thm-restate}
\usepackage{amsmath}
\usepackage[ruled,vlined]{algorithm2e}
\usepackage{cleveref}
\crefname{question}{Question}{Questions}
\usepackage{xcolor}
\usepackage{fullpage}

\DeclareMathOperator{\poly}{poly}
\newcommand{\pathcover}{\mathcal{P}}
\newcommand{\bestparameterizedsolution}{O(k|E|\log{|V|})}
\newcommand{\reachabilitystructure}{\mathcal{S}}
\newcommand{\fptcomplexity}{O(k^24^k|V| + k2^k|E|)}
\newcommand{\reachabilitycomplexity}{O(k2^k)}

\begin{document}
\title{A linear-time parameterized algorithm for computing the width of a DAG\thanks{This work was partially funded by the European Research Council (ERC) under the European Union's Horizon 2020 research and innovation programme (grant agreement No.~851093, SAFEBIO) and by the Academy of Finland (grants No.~322595, 328877).}}
%
%
\author{Manuel C\'aceres\inst{1} \and
Massimo Cairo\inst{1} \and
Brendan Mumey\inst{2} \and
Romeo Rizzi\inst{3} \and
Alexandru I. Tomescu\inst{1}}
\authorrunning{M. C\'aceres et al.}
%
\institute{Department of Computer Science, University of Helsinki, Finland \email{\{manuel.caceresreyes,alexandru.tomescu\}@helsinki.fi, cairomassimo@gmail.com} \and
School of Computer Science, Montana State University, USA \email{brendan.mumey@montana.edu} \and
Department of Computer Science, University of Verona, Italy \email{romeo.rizzi@univr.it}}
\maketitle              
\begin{abstract}
The width $k$ of a directed acyclic graph (DAG) $G = (V, E)$ equals the largest number of pairwise non-reachable vertices. Computing the width dates back to Dilworth's and Fulkerson's results in the 1950s, and is doable in quadratic time in the worst case. Since $k$ can be small in practical applications, research has also studied algorithms whose complexity is parameterized on $k$. Despite these efforts, it is still open whether there exists a \emph{linear-time} $O(f(k)(|V| + |E|))$ parameterized algorithm computing the width. We answer this question affirmatively by presenting an $\fptcomplexity$ time algorithm, based on a new notion of \emph{frontier antichains}. As we process the vertices in a topological order, all frontier antichains can be maintained with the help of several combinatorial properties, paying only $f(k)$ along the way. The fact that the width can be computed by a single $f(k)$-sweep of the DAG is a new surprising insight into this classical problem. Our algorithm also allows deciding whether the DAG has width at most $w$ in time $O(f(\min(w,k))(|V|+|E|))$.

\end{abstract}
%
%
%
%

\section{Introduction}



An \emph{antichain} in a directed acyclic graph (DAG) $G = (V, E)$ is a set of vertices that are pairwise non-reachable. The size $k$ of a maximum-size antichain is also called the \emph{width} of $G$. By Dilworth's theorem~\cite{dilworth2009decomposition}, the width of $G$ also equals the minimum number of paths needed to cover all the vertices of $G$. As such, it can be computed with minimum path cover algorithms e.g., in time $O(\sqrt{|V|} |E^*|)$ by a reduction to maximum matching~\cite{fulkerson1956note,hopcroft1973n} (where $E^*$ is the set of edges in the transitive closure of $G$, and we assume that it is already computed), or in time $O(|V||E|)$ by another reduction to minimum flows~\cite{BJG00,orlin2013max}.

Computing the width of a given DAG has applications in various fields. For example, in distributed computing, it is important to analyze if a distributed program can run so that no more than $w$ processes have mutual access to some resource; this relies on testing whether a particular DAG inferred from of the program trace has width $k \le w$~\cite{ikiz2006efficient,tomlinson1997monitoring}; in bioinformatics, the problems of Perfect Phylogeny Haplotype~\cite{bonizzoni2007linear}, and of Perfect Path Phylogeny Haplotyping~\cite{gramm2007haplotyping} are solved by recognizing special DAGs of width at most two; in evolutionary computation, the so-called dimension of a game between co-evolving agents~\cite{jaskowski2011formal} equals the width of a DAG defined from a minimum coordinate system of the game. For several practical applications, the width of the DAG may be small, for example, in~\cite{makinen2019sparse} the DAG comes from a so-called \emph{pan-genome} encoding genetic variation in a population: this has hundreds of millions of vertices, but yet it has a small width. Furthermore, there exist \emph{fixed-parameter tractable (FPT)} algorithms for several problems on DAGs, which are parameterized by the width of the DAG (see examples in scheduling~\cite{van2016precedence,colbourn1985minimizing} and computational logic~\cite{bova2015model,gajarsky2015fo}), therefore, efficiently recognizing graphs of small width becomes vital for their application. It is thus natural to ask whether there exists a faster algorithm computing the width $k$ of a DAG, when $k$ is small. This question is also related to the line of research ``FPT inside P''~\cite{giannopoulou2017polynomial} of finding natural parameterizations for problems already in P (see also e.g.,~\cite{DBLP:conf/soda/FominLPSW17,koana2021data,abboud2016approximation}).

Along this line, Felsner et~al.~\cite{felsner2003recognition} present the first algorithm parameterized on~$k$, working for the special case of \emph{transitive DAGs}, and running in time $O(k|V|^2)$. They also show how to recognize transitive DAGs of width $2$ and $3$ in time $O(|V|)$, and of width $4$ in time $O(|V|\log{|V|})$. The next parameterized algorithms for general DAGs are due to Chen and Chen: the first runs in time $O(|V|^2 + k\sqrt{k}|V|)$~\cite{chen2008efficient}, and the second one in time $O(\sqrt{|V|}|E| + k\sqrt{k}|V|)$~\cite{chen2014graph}. Recently, M\"akinen et al.~\cite{makinen2019sparse} obtained a faster one for sparse graphs, running in time $\bestparameterizedsolution$. 

Despite these efforts, the time complexity of computing the width of a DAG parameterized on $k$ is not fully settled, since all existing algorithms have either a superlinear dependence on $|E|$, or a quadratic dependence on $|V|$, in the worst case. We present here the first algorithm running in time $O(f(k)(|V| + |E|))$, where $f(k)$ is a function depending only on~$k$. Thus, for constant $k$, this is the first algorithm to run in linear time. Moreover, if an integer $w$ is also given in input, we can decide whether $k \le w$ in time $O(f(\min(w,k))(|V| + |E|))$. Specifically, our main result is the following theorem:

\begin{restatable}{theorem}{fptalgorithm}
\label{thm:main-fpt}
Given a DAG $G = (V,E)$ of width $k$, we can compute a maximum antichain of it in time $\fptcomplexity$.
\end{restatable}
Note that $k$ corresponds to a property of the input graph that is unknown for the algorithm.
\paragraph{Approach.} 

The main idea behind \Cref{thm:main-fpt} is to traverse the graph in a topological order and have an antichain structure sweeping the vertices of the graph, while performing only $f(k)$ work per step. As such, it can also be viewed as an online algorithm receiving in every step a sink vertex and its incoming edges\footnote{Note that this notion of online algorithm is different from the ``on-line chain partition'' problem~\cite{bosek2012line}, where irrevocable decisions opt to be competitive against an optimal solution.}.

As a first attempt to obtain such a ``sweeping'' algorithm, one can think of maintaining only the (unique) \emph{right-most} maximum antichain (recall that all maximum antichains form a lattice~\cite{dilworth1990some})\footnote{Formally, we call right-most maximum antichain to the top element in the lattice of maximum antichains. If the graph is drawn with edges from left to right this element visually corresponds to the right-most maximum antichain.}. However, it is difficult to update this antichain in time $f(k)$ since inherently we need to perform graph traversals. As a second attempt, one could maintain more structure at every step (in addition to the right-most maximum antichain), while still staying within the $f(k)$ budget. Along this line, for transitive DAGs Felsner et~al.~\cite{felsner2003recognition} propose to maintain a \emph{tower} of right-most maximum antichains of decreasing size. That is, take the right-most maximum antichain of $G$, then consider the subgraph strictly reached by this antichain. Then take the right-most maximum antichain of this subgraph, and repeat. One thus obtains a tower of at most $k$ antichains. Felsner et~al.~manage to maintain this structure based on an exhaustive combinatorial approach for $k = 2, 3, 4$, with the former two cases leading to $O(|V|)$ time algorithms, and the latter leading to an $O(|V|\log|V|)$ time algorithm. They also state that ``the case $k=5$ already seems to require an unpleasantly involved case analysis''~\cite[p.~359]{felsner2003recognition}. Moreover, the transitivity of the DAG is crucial in this approach, since reachability between two vertices is equivalent to the existence of an edge between them.

In order to break both of these barriers, we need a different and richer structure to maintain. As such, in \Cref{sec:frontier-antichains} we introduce the notion of \emph{frontier antichain}. A frontier antichain is one such that there is no other antichain \emph{of the same size} and ``to the right'' of it (i.e., no one that \emph{dominates} it, see \Cref{def:frontier-antichain}). Thus, the largest frontier antichain is also the (unique) right-most maximum antichain, and gives the width of $G$. Furthermore, since any antichain can take at most one vertex from any path in a path cover, there are at most $O(2^k)$ frontier antichains (\Cref{lem:at-most-2^k}).

In \Cref{sec:maintaining-f-antichains} we prove several combinatorial properties for maintaining all frontier antichains when a new vertex $v$ in the topological order is added. We show that a frontier antichain of the new graph is either of the form $A \cup \{v\}$, where $A$ is a frontier antichain of the old graph (\Cref{type1,all-type1}), or it is an old frontier antichain that is not dominated by a new frontier antichain (\Cref{lemma:frontiers-dominate,type2}). Thus, it suffices to check domination only between all old and new frontier antichains. However, since domination involves checking reachability (and the DAG is not assumed to be transitive), this might require $O(|V|+|E|)$ time, which we want to avoid. As such, in \Cref{sec:reachability-f-antichains} we prove another key ingredient, namely that it is sufficient to know which vertices in the current frontier antichains reach~$v$ (\Cref{fpt-correctness}). If we maintain this information for every added vertex ($O(k2^k)$ per vertex and edge\footnote{As a purely combinatorial inquiry, we leave open the question of whether the union of all frontier antichains of a given DAG has size $O(\poly(k))$ (instead of $O(k2^k)$).}, \Cref{fpt-reduction}) we can answer the queries required to test domination. Finally, in \Cref{sec:algorithm}, we combine these pieces into the main result if this paper, \Cref{alg:fpt}.

\paragraph{Notation and preliminaries.}
We say that a graph $S = (V_S, E_S)$ is a \emph{subgraph} of $G$ if $V_S \subseteq V$ and $E_S \subseteq E$. If $V' \subseteq V$, then $G[V']$ is the subgraph of $G$ \emph{induced} by $V'$, defined as $G[V'] = (V', E_{V'})$, where $E_{V'} = \{(u, v) \in E ~:~ u,v \in V'\}$. A \emph{path} $P$ is a sequence of different vertices $v_1, \ldots , v_{\ell}$ of $G$ such that $(v_i, v_{i+1}) \in E$, for all $i \in \{1, \ldots, \ell-1\}$. We say that a path $P$ is \emph{proper} if $\ell \geq 2$. A \emph{path cover} $\pathcover$ is a set of paths such that every vertex belongs to some path of $\pathcover$. A \emph{cycle} is a proper path allowed to start and end at the same vertex. A \emph{directed acyclic graph (DAG)} is a graph that does not contain cycles. For a DAG $G = (V, E)$ we can find in $O(|V|+|E|)$ time~\cite{kahn1962topological,tarjan1976edge} an order of its vertices $v_1, \ldots, v_{|V|}$ such that for every edge $(v_i, v_j)$, $i < j$, we call such an order a \emph{topological order}. We say that $v$ is reachable from $u$, or equivalently, that $u$ \emph{reaches} $v$, if there exists a path starting at $u$ and ending at $v$. The problem of efficiently answering whether $u$ reaches $v$ is known as \emph{reachability queries}, and if the queries are answered in constant time, \emph{constant-time reachability queries}. An \emph{antichain} $A$ is a set of vertices such that for each $u,v \in A$ $u\not=v$ $u$, does not reach $v$. We say that $A$ \emph{reaches} a vertex $v$ if there exists $u \in A$ such that $u$ reaches $v$. Dilworth's theorem~\cite{dilworth2009decomposition} states that the maximum size of an antichain equals the minimum size of a path cover in a DAG, this size is known as the \emph{width} of the DAG and denoted by $k$. A \emph{partially ordered set (poset)} is a set $P$ and a \emph{partial order} (reflexive, transitive and antisymmetric binary relation) over $P$. If $P$ is finite, then there exists at least one maximal (minimal) element, and every element in the poset is comparable to some maximal (minimal) element~\cite{wallis2016introduction}. A \emph{maximal (minimal)} element of a poset is an element that is not smaller (greater) than any other element.

\section{Frontier Antichains} \label{sec:frontier-antichains}
 
 We start by introducing the concept of \emph{frontier antichains} and show a bound on the number of frontier antichains present in a DAG.
  \begin{definition}[Antichain domination]
    \label{def:antichain-domination}
     Let $A$ and $B$ be antichains of the same size.
     We say that $B$ \emph{dominates} $A$ if for all $b \in B$, $A$ reaches $b$.
 \end{definition}
 
 \begin{algorithm}[t]
    \DontPrintSemicolon
    \SetKwProg{Fn}{Function}{:}{}
    \Fn{$dominates(B, A, \reachabilitystructure)$}{
        $isDominated \gets$ $|A| = |B|$\tcp*[r]{\texttt{true} if $A$ is dominated by $B$}
        \For{$v \in B$}{
            \If(\tcp*[f]{see \Cref{alg:reaches}}){\texttt{not} $reaches(A, v, \reachabilitystructure)$}{
                $isDominated \gets$ \texttt{false}\;
            }
        }
        \Return $isDominated$\;
    }
    \caption{\label{alg:dominates}Function $dominates(B, A, \reachabilitystructure)$ checks if an antichain $B$ dominates an antichain $A$ (\Cref{def:antichain-domination}), assuming that the structure $\reachabilitystructure$ can compute reachability from vertices of $A$ to vertices of $B$. If $reaches(A, v, \reachabilitystructure)$ takes $O(|A|)$ time (see \Cref{alg:reaches}), then this function takes $O(|A||B|) = O(k^2)$ time.}
\end{algorithm}

 
Note that antichains can only dominate other antichains of the same size, since antichains of different size are, by definition, incomparable. \Cref{alg:dominates} shows a function determining whether an antichain dominates another. The following lemma shows that the set of antichains of $G$ with the domination relation form a partial order.
\begin{lemma}
\label{lemma:domination-partial-order}
The antichains of $G$ related with domination (\Cref{def:antichain-domination}) form a partial order.
\end{lemma}
\begin{proof}
     Clearly, domination is reflexive and transitive (inherited by the transitivity of reachability between vertices).  We argue that it is also antisymmetric: suppose $A$ and $B$ are antichains such that $A$ dominates $B$ and $B$ dominates $A$. Suppose by contradiction that there exists $b \in B \setminus A$. Since $B$ dominates $A$, there exists $a \in A$ such that $a$ reaches $b$ (note $a \ne b$). Since $A$ dominates $B$, there exists $b' \in B$ such that $b'$ reaches $a$. Thus, there is a proper path from $b'$ to $a$ to $b$ in $G$.  If $b' = b$, this implies a cycle exists in a DAG, a contradiction.  If $b' \ne b$, this implies $B$ is not an antichain, a contradiction. Thus, $B \setminus A = \emptyset$ and $A=B$, since $|A|=|B|$. Thus, domination is also antisymmetric.
\qed\end{proof}

 \begin{definition}[Frontier antichains]
     \label{def:frontier-antichain}
     Frontier antichains are the maximal elements of the domination partial order i.e., those antichains that are not dominated by any other antichain.
 \end{definition}

\begin{figure}[t]
    \centering
    \includegraphics[scale=0.7]{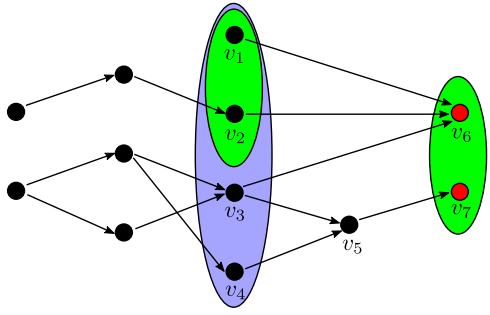}
    \caption{A DAG and all its frontier antichains of size $1, 2$ and $4$, as colored sets. The sub-indices represent a topological order. The unique maximum-size frontier antichain is $\{v_1, v_2, v_3, v_4\}$, and is also the right-most maximum antichain. There are $2$ frontier antichains of size $2$, $\{v_1, v_2\}$ and $\{v_6, v_7\}$. The frontier antichains of size $1$ are $\{v_6\}$ and $\{v_7\}$. Frontier antichains of size $3$ are not highlighted, there are $3$ of them, $\{v_1, v_2, v_7\}$, $\{v_1, v_3, v_4\}$ and $\{v_2, v_3, v_4\}$.}
    \label{fig:frontier-antichains}
\end{figure}

\Cref{fig:frontier-antichains} shows frontier antichains of an example graph. The next lemma establishes that frontier antichains dominate all antichains of the graph i.e., every  non-frontier antichain is dominated by some frontier antichain (thus of the same size).

 \begin{lemma}\label{lemma:frontiers-dominate}
Let $A$ be a  non-frontier antichain of $G$. Then, there exists a frontier antichain dominating $A$.
\end{lemma}
\begin{proof}
    Since there are a finite number of antichains of $G$, the antichains with the domination relation form a finite poset (\Cref{lemma:domination-partial-order}), therefore every element of this poset (i.e.,~antichain) is less than or equal to (i.e.,~is dominated by) a maximal element (i.e.,~a frontier antichain).
\qed\end{proof}

Now we show that the number of such antichains only grows with $k$, thus there is no problem for our complexity bound to maintain them all. The following lemma shows that there are at most $2^k$ frontier antichains. The main idea is that there cannot be more than one frontier antichain whose vertices belong to the same set of paths in a minimum path cover of $G$.
\begin{lemma}
    \label{lem:at-most-2^k}
    If $G = (V, E)$ is a DAG of width $k$, then $G$ has at most $2^k$ frontier antichains.
\end{lemma}
\begin{proof}
    By Dilworth's theorem~\cite{dilworth2009decomposition}, there exists a path cover of $G$ of size $k$, $\pathcover = \{P_1, \ldots, P_k\}$. Since any antichain can take at most one vertex from each of those paths, we show that for every size-$\ell$ subset of paths of $\pathcover$, there is at most one frontier antichain of size $\ell$ whose vertices come from those paths, and thus there are at most $2^k$ frontier antichains. Without loss of generality consider the subset of paths $P_1, \ldots, P_{\ell}$, and suppose by contradiction that there are two frontier antichains $A$ and $B$, $A \not = B$, $|A|=|B|=\ell$, whose vertices come from $P_1, \ldots, P_{\ell}$. Let us label the vertices in these antichains by the path they belong to. Namely, $A = \{a_1,\dots,a_{\ell}\}$, $B = \{b_1,\dots,b_{\ell}\}$, with $a_i$ and $b_i$ in $P_i$ for all $i \in \{1,\ldots, \ell\}$. We define the following set of vertices:
    \begin{align*}
        M :=  \left\{m_i := \left(b_i,\text{ if } a_i\text{ reaches }b_i\text{, and } a_i \text{ otherwise}\right) ~\mid~ i \in \{1,\ldots, \ell\}\right\}.
    \end{align*}
    First, note that if $m_i = a_i$, then $b_i$ reaches $a_i$, because $a_i$ and $b_i$ appear on the same path $P_i$. Next, note that $M$ is an antichain of size $\ell$. Otherwise, if there exists $m_i$ that reaches $m_j$ ($i \neq j)$, then without loss of generality, suppose that $m_i = a_i$ and $m_j = b_j$. Since $m_i = a_i$, we have that $b_i$ reaches $a_i$, and thus it reaches $b_j$, which contradicts $B$ being an antichain. Second, note that $M \neq A$, since otherwise $A$ would dominate $B$. Finally, $M$ dominates $A$, since for all $m_i \in M$ there exists $a_i\in A$ such that $a_i$ reaches $m_i$.
\qed\end{proof}

\section{Maintaining frontier antichains}
\label{sec:maintaining-f-antichains}

Our algorithm will process the vertices in topological order $v_1,\ldots, v_{|V|}$, and maintain all \emph{frontier} antichains (\Cref{def:frontier-antichain}) of the current subgraph $G_i := G[\{v_1,\ldots, v_i\}]$ (we say that $G_0 = (\emptyset, \emptyset)$). The following property allows us to upper bound the width of each of these induced subgraphs by the width $k$ of the original graph.

\begin{property}
    \label{topologicalDoNotIncreaseWidth}
    Let $G = (V, E)$ be a DAG of width $k$, and $v_1, \ldots , v_{|V|}$ a topological order of its vertices. Then, for all $i, j \in \{1,\ldots, |V|\}, i \le j$, the width of $G_{i,j} := G[\{v_{i}, \ldots , v_{j}\}]$ is at most $k$.
\end{property}
\begin{proof}
     We first show that the intersection of any path of $G$ with the vertices of $G_{i, j}$ is a path in $G_{i, j}$. Consider a path $P$, and remove from it all the vertices from $V\setminus \{v_i, \ldots, v_j\}$. Thus, we obtain a (possibly empty) sequence $P_{i, j}$ of vertices from $\{v_i, \ldots, v_j\}$. We say that $P_{i, j}$ is the \emph{intersection} of $P$ with $G_{i, j}$. Since $G$ is a DAG, $P_{i, j}$ is a sequence of consecutive vertices in $P$ (otherwise, if it is not empty, we would have a vertex of smaller (bigger) topological index that is reached by $v_i$ (reaches $v_j$)), and therefore a path in $G$. Since $P_{i, j}$ only contains vertices from $\{v_i, \ldots, v_j\}$ and $G_{i, j}$ is an induced subgraph, $P_{i, j}$ is a path also in $G_{i, j}$.

     By Dilworth's theorem~\cite{dilworth2009decomposition}, there exists a path cover of $G$ of size $k$, $\pathcover = \{P_1, \ldots, P_k\}$. The intersection of each of those paths with $G_{i, j}$ forms a path cover of $G_{i, j}$, whose size is at least the width of $G_{i, j}$.
\qed\end{proof}

We say that an antichain is \emph{$G_i$-frontier} if it is a frontier antichain in the graph $G_i$. The following two lemmas will show us how these frontier antichains evolve when processing the vertices of the graph i.e., when passing from $G_{i-1}$ to $G_{i}$.

\begin{lemma}[Type 1]
 \label{type1}
 For every $i \in \{1,\ldots, |V|\}$, let $A$ be a $G_i$-frontier antichain with $v_i \in A$. Then $A\setminus \{v_i\}$ is a $G_{i-1}$-frontier antichain.
\end{lemma}
\begin{proof}
    Otherwise, there would exist another antichain $B$ dominating $A\setminus \{v_i\}$ in $G_{i-1}$. Consider $B \cup \{v_i\}$, which is an antichain (otherwise $B$ would reach $v_i$, a contradiction, since $B$ dominates $A\setminus \{v_i\}$, and $A$ is an antichain). Finally, note that $B \cup \{v_i\}$ dominates $A$ in $G_i$, which is a contradiction since $A$ is $G_i$-frontier antichain.
\qed\end{proof}

\begin{lemma}[Type 2]
\label{type2}
 For every $i \in \{1,\ldots, |V|\}$, let $A$ be a $G_i$-frontier antichain with $v_i \not \in A$. Then $A$ is a $G_{i-1}$-frontier antichain.
\end{lemma}
\begin{proof}
    Otherwise, there would exist another antichain $B$ dominating $A$ in $G_{i-1}$, and also in $G_i$, which is a contradiction.
\qed\end{proof}

Looking at these two lemmas, we establish two types of $G_i$-frontier antichains: the ones containing $v_i$, called of \emph{type 1}, and the ones that are also $G_{i-1}$-frontier antichains, called of \emph{type 2}. We handle these two cases separately. First, we find all type-1 frontier antichains, then all of type~2.

Type-1 $G_i$-frontier antichains are made up of one $G_{i-1}$-frontier antichain and vertex $v_i$. A first requirement for a $G_{i-1}$-frontier antichain, $A$, to be a subset of a type-1 $G_{i}$-frontier antichain is that $A$ does not reach $v_i$. We now show that this is enough to ensure that $A \cup \{v_i\}$ is a $G_i$-frontier antichain.

\begin{lemma}
\label{all-type1}
 For every $i \in \{1,\ldots, |V|\}$, let $A$ be a $G_{i-1}$-frontier antichain not reaching $v_i$. Then $A \cup \{v_i\}$ is a $G_i$-frontier antichain.
\end{lemma}
\begin{proof}
    If $A = \emptyset$, then $A \cup \{v_i\} = \{v_i\}$ is frontier antichain, because $v_i$ is a sink of $G_i$. Otherwise $A \ne \emptyset$ and, by contradiction, take another antichain $B$ dominating $A \cup \{v_i\}$ in $G_i$. Suppose that $v_i\in B$, then for all $b \in B$ there exists $a \in A \cup \{v_i\}$ such that $a$ reaches $b$, but since $v_i$ is a sink of $G_i$, for all $b \in B\setminus\{v_i\}$ there exists $a \in A$ such that $a$ reaches $b$ i.e., $B\setminus\{v_i\}$ dominates $A$ in $G_{i-1}$, which is a contradiction. If $v_i \not \in B$, then every vertex of $B$ is reached by a vertex of $A$ (it cannot be reached by $v_i$ since it is a sink in $G_i$), and therefore take any subset of $B$ of size $|A|$ different from $A$, which would dominate $A$, a contradiction.
\qed\end{proof}

We use this lemma to find all type-1 $G_i$-frontier antichains by testing reachability from $G_{i-1}$-frontier antichains to $v_i$, with $O(k2^k)$ reachability queries in total.

Type-2 $G_i$-frontier antichains are $G_{i-1}$-frontier antichains that are not dominated by any antichain in $G_i$ containing $v_i$ (this is sufficient since they are frontier in $G_{i-1}$). Moreover, by \Cref{lemma:frontiers-dominate}, if a $G_{i-1}$-frontier antichain is dominated in $G_i$, then it is dominated by a $G_i$-frontier antichain. Therefore, type-2 $G_i$-frontier antichains are $G_{i-1}$-frontier antichains that are not dominated by any type-1 $G_i$-frontier antichain. For every $G_{i-1}$-frontier antichain $A$ we check if there exists a type-1 $G_i$-frontier antichain dominating $A$. We can do this in total $O(k^24^k)$ reachability queries from vertices in $G_{i-1}$-frontier antichains to vertices in $G_{i-1}$-frontier antichains and $v_{i}$. 

Both type-1 and type-2 $G_i$-frontier antichains need answering reachability queries efficiently among vertices in $G_{i-1}$-frontier antichains and $v_i$. Next, we show how to maintain constant-time reachability queries among these vertices in $\reachabilitycomplexity$ time per vertex and edge.

\section{Reachability between frontier antichains}
\label{sec:reachability-f-antichains}

To complete our algorithm, we aim to maintain reachability queries among all vertices in $G_{i-1}$-frontier antichains and $v_i$. For this we rely on properties of the support of the frontier antichains, as detailed next.

\begin{definition}[Support]
    For every $i \in \{0, 1,\ldots, |V|\}$, we define the \emph{support} $S_i$ of $G_{i}$ as the set of all vertices belonging to some $G_{i}$-frontier antichain, that is,
    \begin{align*}
        S_i := \bigcup_{\text{$A$~:~$G_i$-frontier\ antichain}} A.
    \end{align*}
\end{definition}

Note that since $G_0 = (\emptyset, \emptyset)$, then $S_0 = \emptyset$, and \Cref{lem:at-most-2^k} implies $|S_i| = O(k2^k)$. Also, $v_i \in S_i$, since $\{v_i\}$ is a $G_i$-frontier antichain. In \Cref{fig:frontier-antichains}, the vertex $v_5$ belongs to the support of $S_5$, but it does not belong to $S_7$, because there is no frontier antichain containing it. Another interesting fact is that if a vertex exits the support in some step, then it cannot re-enter. This is formalized as follows.

\begin{lemma}
\label{staticsupport}
 Let $v \in \{v_1, \ldots, v_i\}$. If $v \not \in S_i$, then $v\not \in S_j$ for all $j\in\{i,\ldots, |V|\}$.
\end{lemma}
\begin{proof}
    By induction on $j$. The base case $j = i$ is the hypothesis itself. Now, suppose that $v\not \in S_j$ for some $j \in \{i, \ldots, |V|-1\}$, and suppose by contradiction that $v \in S_{j+1}$. Then $v \in A$, for some $G_{j+1}$-frontier antichain $A$. If $v_{j+1} \not \in A$, then by \Cref{type2}, $A$ is a $G_{j}$-frontier antichain, and $v \in S_{j}$, which is a contradiction. But if $v_{j+1} \in A$, then by \Cref{type1}, $A\setminus \{v_{j+1}\}$ is a $G_{j}$-frontier antichain, and $v \in A\setminus \{v_{j+1}\} \subseteq S_j$, a contradiction.
\qed\end{proof}

\begin{lemma}
\label{staticsupport2}
 Let $v_i \in \{v_1, \ldots, v_j\}$. If $v_i \in S_{j}$, then $v_i \in S_t$ holds for all $t\in \{i,\ldots, j\}$.
\end{lemma}
\begin{proof}
    If this is not true, we have that there exists some $t \in \{i+1, \ldots, j-1\}$ such that $v_i \not \in S_t$, which is a contradiction with $v_{i} \in S_j$ and \Cref{staticsupport}.
\qed\end{proof}

We now state that it is sufficient to support reachability queries from every $S_{j-1}$ to $v_j$ to answer queries among vertices in $S_{i-1}$ and $v_i$. Then, we show how to maintain these reachability relations in $\reachabilitycomplexity$ time per vertex and edge.

\begin{theorem}
    \label{fpt-reduction}
    If we know reachability from $S_{j-1}$ to $v_j$ for all $j\in[1...i]$, then we can answer reachability queries among vertices in $S_{i-1} \cup \{v_i\}$.
\end{theorem}
\begin{proof}
    Let $v_s, v_t \in S_{i-1} \cup \{v_i\}$. We can answer whether $v_s$ reaches $v_t$ by doing the following. If $s \ge t$ it is not possible that $v_s$ reaches $v_t$ unless they are the same vertex. In the other case, $s < t$, since $v_s\in S_{i-1}$, by \Cref{staticsupport2}, $v_s \in S_{t-1}$, and then we can use reachability from $S_{t-1}$ to $v_t$ to answer this query.
\qed\end{proof}

\begin{algorithm}[t]
    \DontPrintSemicolon
    \SetKwProg{Fn}{Function}{:}{}
    \Fn{$reaches(A, v_t, \reachabilitystructure = (S_{0}, \ldots, S_{i-1}))$}{
        $isReached \gets$ \texttt{false}\tcp*[r]{\texttt{true} if $v_t$ is reached from some vertex in $A$}
        \For{$v_s \in A$}{
            \If{$v_s = v_t$ or ($s < t$ and $S_{t-1}.v_s.reaches$)}{
                $isReached \gets$ \texttt{true}\;
            }
        }
        \Return $isReached$\;
    }
    \caption{\label{alg:reaches}Function $reaches(A, v_t, \reachabilitystructure)$, with $\reachabilitystructure = (S_{0}, \ldots, S_{i-1})$ for some $i \ge t$, and $A\cup\{v_t\}\subseteq S_{i-1}\cup\{v_i\}$. It checks if $A$ reaches $v_t$. It assumes that for all the vertices $u \in S_{t-1}$, $S_{t-1}.u.reaches$ indicates if $u$ reaches $v_t$. Reachability in $S_{i-1}\cup \{v_i\}$ is reduced to reachability from $S_{j-1}$ to $v_j$ for all $j \in \{1,\ldots, i\}$, according to \Cref{fpt-reduction}. This function takes $O(|A|) = O(k)$ time.}
\end{algorithm}

\Cref{alg:reaches} shows a function deciding whether an antichain reaches a vertex, using the technique explained in \Cref{fpt-reduction}. This function is used to implement \Cref{alg:dominates}, and our final solution in \Cref{alg:fpt}.

We will compute reachability from $S_{j-1}$ to $v_j$ for all $j \in \{1,\ldots, i\}$ incrementally when processing the vertices in topological order. That is, we assume that we have computed reachability from $S_{j-1}$ to $v_j$ for all $j \in \{1,\ldots, i-1\}$ and we want to compute reachability from $S_{i-1}$ to $v_i$.

For this we do the following. Initially, we set reachability from $u$ to $v_i$ to \texttt{false} for all $u \in S_{i-1}$. Then, for every edge $(v_j, v_i)$, if $v_j \in S_{i-1}$ we set reachability from $v_j$ to $v_i$ to \texttt{true}, and for each $u \in S_{i-1} \cap S_{j-1}$ such that $u$ reaches $v_j$ (known since $u \in S_{j-1}$) we set reachability from $u$ to $v_i$ to \texttt{true}. Note that we can compute the intersection $S_{i-1} \cap S_{j-1}$ in $O(|S_{i-1}|) = \reachabilitycomplexity$ time. For each $v_p \in S_{i-1}$ we decide whether $v_p\in S_{j-1}$ by testing if $p \le j-1$, which is correct by \Cref{staticsupport2}.

\begin{algorithm}[t]
    \DontPrintSemicolon
    \SetKwProg{Fn}{Function}{:}{}
    \Fn{$updateReachability(v_i, \reachabilitystructure = (S_{0}, \ldots, S_{i-1}))$}{
        \For{$u \in S_{i-1}$}{
            $S_{i-1}.u.reaches \gets$ \texttt{false}\tcp*[r]{\texttt{true} if $u$ reaches $v_i$}
        }
        \For{$v_j \in N^-(v_i)$}{
            \If(\tcp*[f]{Direct (by one edge) reachability}){$v_j \in S_{i-1}$}{
                $S_{i-1}.v_j.reaches \gets$ \texttt{true}\;
            }
            \For(\tcp*[f]{More than one edge reachability}){$u \in S_{i-1} \cap S_{j-1}$}{
                \If{$S_{j-1}.u.reaches$}{
                    $S_{i-1}.u.reaches \gets$ \texttt{true}\;
                }
            }
        }
    }
    \caption{\label{alg:update}Function $updateReachability$ computes reachability from vertices in $S_{i-1}$ to $v_i$. It assumes that $S_{i-1}\in\reachabilitystructure$, for all $j\in \{1, \ldots, i-1\}$, $S_{j-1} \in \reachabilitystructure$ , and for all the vertices $u \in S_{j-1}$, $S_{j-1}.u.reaches$ indicates if $u$ reaches $v_j$. Correctness of this function is explained in \Cref{fpt-correctness}. This function takes $O(k2^k(|N^-(v_i)| + 1))$ time.}
\end{algorithm}

\Cref{alg:update} shows a function that computes the reachability from $S_{i-1}$ to $v_i$, according to what was explained in this section. The correctness of this procedure is guaranteed by the following theorem.

\begin{theorem}
    \label{fpt-correctness}
    \Cref{alg:update} computes reachability from $S_{i-1}$ to $v_i$.
\end{theorem}
\begin{proof}
    Clearly, what the algorithm sets to \texttt{true} is correct. Suppose by contradiction that there exists $y \in S_{i-1}$ reaching $v_i$ such that reachability from $y$ to $v_{i}$ was not set to \texttt{true}. Since $y$ reaches $v_{i}$, the in-neighborhood of $v_{i}$ is not empty. Since $y$ was not set to \texttt{true}, in particular, $y \not \in N^-(v_i)$, thus it reaches $v_{i}$ through a path whose last vertex previous to $v_i$ is $v_j \in N^-(v_i)$. Again, since $y$ was not set to \texttt{true}, $y \not \in S_{i-1} \cap S_{j-1}$, thus $y \not \in S_{j-1}$. But then, by \Cref{staticsupport} we have $y \not \in S_{i-1}$, a contradiction, unless $y \not \in G_{j-1}$ i.e., $y$ is after $v_j$ in topological order, which is a contradiction since $y$ reaches $v_j$.
\qed\end{proof}

\section{A linear-time parameterized algorithm}\label{sec:algorithm}
We now have all the ingredients to prove the main theorem.

\fptalgorithm*

\begin{proof}
    We process the vertices in topological order. After processing $v_i$, we will have computed all $G_i$-frontier antichains (including the right-most maximum antichain of $G_i$), and constant-time reachability queries from $S_{j-1}$ to $v_j$, for all $j\in\{1,\ldots, i\}$. Suppose we have this for $i-1$.\footnote{Since $G_0 = (\emptyset, \emptyset)$ and $S_0 = \emptyset$, there are no frontier antichains for the base case of the algorithm.} First, we obtain constant-time reachability queries from $S_{i-1}$ to $v_i$, using the procedure from \Cref{fpt-correctness} (\Cref{alg:update}), spending $\reachabilitycomplexity$ time, and $\reachabilitycomplexity$ time per edge incoming to $v_i$. For the entire algorithm, this adds up to $O(k2^k(|V| + |E|))$.
    
    By \Cref{all-type1}, we obtain all type-1 $G_i$-frontier antichains by taking every $G_{i-1}$-frontier antichain $A$, and testing if $A$ reaches $v_i$ using the reduction from \Cref{fpt-reduction} (\Cref{alg:reaches}). This takes $O(k2^k)$ time,  $O(k2^k|V|)$ in total.
    
    We compute type-2 $G_i$-frontier antichains by taking every $G_{i-1}$-frontier antichain $A$ and searching if there exists a type-1 $G_i$-frontier antichain $B$ dominating~$A$ in time $O(k^24^k)$ ($O(k^2)$ constant-time reachability queries to test domination between a pair of antichains, $O(4^k)$ such pairs), $O(k^24^k|V|)$ in total. The total complexity of the algorithm is $O(k^24^k|V| + k2^k|E|)$.
\qed\end{proof}

\begin{algorithm}[t]
\DontPrintSemicolon
    $R \gets \emptyset$, $\mathcal{F}_{0} \gets \{\emptyset\}$, $\reachabilitystructure \gets (S_0=\emptyset)$\;
    \For{$v_i \in v_1, \ldots, v_{|V|}$ in topological order}{
        $updateReachability(v_i, \reachabilitystructure)$\;
        $\mathcal{F}_i \gets \{\emptyset\}$\tcp*[r]{$\mathcal{F}_i$ stores $G_i$-frontiers}
        \For(\tcp*[f]{Compute type-1 $G_{i}$-frontiers}){antichain $A \in \mathcal{F}_{i-1}$}{
            \If{\texttt{not} $reaches(A, v_i, \reachabilitystructure)$}{
                $\mathcal{F}_i.add(A \cup \{v_i\})$\tcp*[r]{$A \cup \{v_i\}$ is a type-1 $G_{i}$-frontier}
                \lIf{$|A \cup \{v_i\}| > |R|$}{
                    $R \gets A \cup \{v_i\}$
                }
            }
        }
        $\mathcal{T}_1 \gets \mathcal{F}_i$\tcp*[r]{Contains type-1 $G_{i}$-frontiers}
        \For(\tcp*[f]{Compute type-2 $G_{i}$-frontiers}){antichain $A \in \mathcal{F}_{i-1}$}{
            $isType2 \gets$ \texttt{true}\tcp*[r]{\texttt{true} if $A$ is a type-2 $G_{i}$-frontier}
            \For{antichain $B \in \mathcal{T}_1$}{
                \lIf{$dominates(B, A, \reachabilitystructure)$}{
                     $isType2 \gets$ \texttt{false}
                }
            }
            \lIf{$isType2$}{
                 $\mathcal{F}_i.add(A)$
            }
        }
        $\reachabilitystructure.add\left(S_i \gets \bigcup_{A \in \mathcal{F}_i}A\right)$\;
    }
    \Return $R$\;
 \caption{\label{alg:fpt}The parameterized algorithm from \Cref{thm:main-fpt} computing the right-most maximum antichain $R$ of size $k$ of a DAG $G = (V, E)$ in time $\fptcomplexity$. Here, $\reachabilitystructure$ is a data structure containing reachability from the previous support to the newly added vertex at each step;  $updateReachability(v_i, \reachabilitystructure)$ computes reachability from vertices in $S_{i-1}$ to $v_i$; $reaches(A, v_i, \reachabilitystructure)$ checks if some vertex of $A$ reaches $v_i$; and $dominates(B, A, \reachabilitystructure)$ checks if an antichain $B$ dominates an antichain $A$ (\Cref{alg:update,alg:reaches,alg:dominates}).}
\end{algorithm}

\Cref{alg:fpt} shows the pseudocode of the final solution explained in \Cref{thm:main-fpt}. It maintains reachability from the corresponding support to the newly added vertex using \Cref{alg:update}. Type-1 frontier antichains are found by using \Cref{alg:reaches}, and type-2 frontier antichains are confirmed using \Cref{alg:dominates}.

Furthermore, the following remark shows that the time complexity of our algorithm can be refined to $O(k^2f^2|V| + kf|E|)$ time, where $f$ is the largest number of frontier antichains encountered at any step. This value can be as much as $2^k$ and as little as $k$.

\begin{remark}\label{remark:further-parameterization}
The number of frontier antichains can be as much as $2^k$ (e.g., $k$ independent vertices), and as little as $k$ (e.g., a sequence of sets of independent vertices of sizes $k, k-1, \ldots, 1$ such that the out-neighborhood of every vertex in the set of size $i$ is the set of size $i-1$). Since in practical examples the number of frontier antichains could be much smaller than its bound $2^k$, we refine the analysis of the algorithm in terms of the number of frontier antichains. Let $F_i$ be the number of frontier antichains in $G_i$. Then the $i$-th step of the algorithm takes $O(|S_{i-1}| + k^2F_{i-1}^2 + kF_i|S_i|)$ time, and $O(|S_{i-1}|)$ time per incoming edge. Noting that $|S_i| = O(kF_i)$, this is $O(k^2F_{i-1}^2 + k^2F_i^2)$ time and $O(kF_{i-1})$ time per incoming edge. If we take $F = \max_{i\in \{1,\ldots, |V|\}} F_i$, then the algorithm takes $O(k^2F^2|V| + kF|E|)$ time.
\end{remark}

Finally, if we are interested in recognizing whether $G$ has width at most an additional input integer $w$ we can adapt our algorithm to run in time $O(f(\min(w, k))$ $(|V| +|E|))$ instead.
\begin{remark}\label{remark:recognition-algorithm}
Given an additional input integer $w$ we can determine whether $k\le w$ in time $O(w'^24^{w'}|V| + w'2^{w'}|E|)$ ($w' = \min(w, k)$) by stopping the computation of \Cref{alg:fpt} as soon as we find an antichain of size $w+1$. If the algorithm does not stop by this reason, it means that $k \le w$, and the opposite otherwise. In both cases maximum size of an observed antichain is not greater than $w'+1$, obtaining the desired running time.
\end{remark}





%
%
%
%

\bibliographystyle{splncs04}
\bibliography{references}

\begin{thebibliography}{10}
\providecommand{\url}[1]{\texttt{#1}}
\providecommand{\urlprefix}{URL }
\providecommand{\doi}[1]{https://doi.org/#1}

\bibitem{abboud2016approximation}
Abboud, A., Williams, V.V., Wang, J.: Approximation and fixed parameter
  subquadratic algorithms for radius and diameter in sparse graphs. In:
  Proceedings of the twenty-seventh annual ACM-SIAM symposium on Discrete
  Algorithms. pp. 377--391. SIAM (2016)

\bibitem{BJG00}
Bang-Jensen, J., Gutin, G.: Digraphs {T}heory, {A}lgorithms and {A}pplications.
  Springer-Verlag, Berlin, 1st edn. (2000)

\bibitem{bonizzoni2007linear}
Bonizzoni, P.: A linear-time algorithm for the perfect phylogeny haplotype
  problem. Algorithmica  \textbf{48}(3),  267--285 (2007)

\bibitem{bosek2012line}
Bosek, B., Felsner, S., Kloch, K., Krawczyk, T., Matecki, G., Micek, P.:
  On-line chain partitions of orders: a survey. Order  \textbf{29}(1),  49--73
  (2012)

\bibitem{bova2015model}
Bova, S., Ganian, R., Szeider, S.: Model checking existential logic on
  partially ordered sets. ACM Transactions on Computational Logic (TOCL)
  \textbf{17}(2),  1--35 (2015)

\bibitem{chen2008efficient}
Chen, Y., Chen, Y.: An efficient algorithm for answering graph reachability
  queries. In: 2008 IEEE 24th International Conference on Data Engineering. pp.
  893--902. IEEE (2008)

\bibitem{chen2014graph}
Chen, Y., Chen, Y.: On the graph decomposition. In: 2014 IEEE Fourth
  International Conference on Big Data and Cloud Computing. pp. 777--784. IEEE
  (2014)

\bibitem{colbourn1985minimizing}
Colbourn, C.J., Pulleyblank, W.R.: Minimizing setups in ordered sets of fixed
  width. Order  \textbf{1}(3),  225--229 (1985)

\bibitem{dilworth2009decomposition}
Dilworth, R.P.: A {D}ecomposition {T}heorem for {P}artially {O}rdered {S}ets.
  Annals of Mathematics  \textbf{51}(1),  161--166 (1950),
  \url{http://www.jstor.org/stable/1969503}

\bibitem{dilworth1990some}
Dilworth, R.P.: Some combinatorial problems on partially ordered sets. In: The
  Dilworth Theorems, pp. 13--18. Springer (1990)

\bibitem{felsner2003recognition}
Felsner, S., Raghavan, V., Spinrad, J.: Recognition algorithms for orders of
  small width and graphs of small {Dilworth} number. Order  \textbf{20}(4),
  351--364 (2003)

\bibitem{DBLP:conf/soda/FominLPSW17}
Fomin, F.V., Lokshtanov, D., Pilipczuk, M., Saurabh, S., Wrochna, M.: Fully
  polynomial-time parameterized computations for graphs and matrices of low
  treewidth. In: Klein, P.N. (ed.) Proceedings of the Twenty-Eighth Annual
  {ACM-SIAM} Symposium on Discrete Algorithms, {SODA} 2017, Barcelona, Spain,
  Hotel Porta Fira, January 16-19. pp. 1419--1432. {SIAM} (2017).
  \doi{10.1137/1.9781611974782.92},
  \url{https://doi.org/10.1137/1.9781611974782.92}

\bibitem{fulkerson1956note}
Fulkerson, D.R.: Note on {Dilworth’s} decomposition theorem for partially
  ordered sets. In: Proc. Amer. Math. Soc. vol.~7, pp. 701--702 (1956)

\bibitem{gajarsky2015fo}
Gajarsk{\`y}, J., Hlinen{\`y}, P., Lokshtanov, D., Obdralek, J., Ordyniak, S.,
  Ramanujan, M., Saurabh, S.: {FO} model checking on posets of bounded width.
  In: 2015 IEEE 56th Annual Symposium on Foundations of Computer Science. pp.
  963--974. IEEE (2015)

\bibitem{giannopoulou2017polynomial}
Giannopoulou, A.C., Mertzios, G.B., Niedermeier, R.: Polynomial fixed-parameter
  algorithms: {A} case study for longest path on interval graphs. Theoretical
  computer science  \textbf{689},  67--95 (2017)

\bibitem{gramm2007haplotyping}
Gramm, J., Nierhoff, T., Sharan, R., Tantau, T.: Haplotyping with missing data
  via perfect path phylogenies. Discrete Applied Mathematics
  \textbf{155}(6-7),  788--805 (2007)

\bibitem{hopcroft1973n}
Hopcroft, J.E., Karp, R.M.: An $n^{5/2}$ algorithm for maximum matchings in
  bipartite graphs. SIAM Journal on computing  \textbf{2}(4),  225--231 (1973)

\bibitem{ikiz2006efficient}
Ikiz, S., Garg, V.K.: Efficient incremental optimal chain partition of
  distributed program traces. In: 26th IEEE International Conference on
  Distributed Computing Systems (ICDCS'06). pp. 18--18. IEEE (2006)

\bibitem{jaskowski2011formal}
Ja{\'s}kowski, W., Krawiec, K.: Formal analysis, hardness, and algorithms for
  extracting internal structure of test-based problems. Evolutionary
  computation  \textbf{19}(4),  639--671 (2011)

\bibitem{kahn1962topological}
Kahn, A.B.: Topological sorting of large networks. Communications of the ACM
  \textbf{5}(11),  558--562 (1962)

\bibitem{koana2021data}
Koana, T., Korenwein, V., Nichterlein, A., Niedermeier, R., Zschoche, P.: Data
  {R}eduction for {M}aximum {M}atching on {R}eal-{W}orld {G}raphs: {T}heory and
  {E}xperiments. Journal of Experimental Algorithmics (JEA)  \textbf{26},
  1--30 (2021)

\bibitem{makinen2019sparse}
M{\"a}kinen, V., Tomescu, A.I., Kuosmanen, A., Paavilainen, T., Gagie, T.,
  Chikhi, R.: {Sparse Dynamic Programming on DAGs with Small Width}. ACM
  Transactions on Algorithms (TALG)  \textbf{15}(2),  1--21 (2019)

\bibitem{orlin2013max}
Orlin, J.B.: {Max flows in $O(nm)$ time, or better}. In: Proceedings of the
  forty-fifth annual ACM symposium on Theory of computing. pp. 765--774 (2013)

\bibitem{tarjan1976edge}
Tarjan, R.E.: Edge-disjoint spanning trees and depth-first search. Acta
  Informatica  \textbf{6}(2),  171--185 (1976)

\bibitem{tomlinson1997monitoring}
Tomlinson, A.I., Garg, V.K.: Monitoring functions on global states of
  distributed programs. Journal of Parallel and Distributed Computing
  \textbf{41}(2),  173--189 (1997)

\bibitem{van2016precedence}
Van~Bevern, R., Bredereck, R., Bulteau, L., Komusiewicz, C., Talmon, N.,
  Woeginger, G.J.: Precedence-constrained scheduling problems parameterized by
  partial order width. In: International conference on discrete optimization
  and operations research. pp. 105--120. Springer (2016)

\bibitem{wallis2016introduction}
Wallis, W.D., George, J.C.: Introduction to combinatorics. CRC press (2016)

\end{thebibliography}

\end{document}